\documentclass[numbib]{scrartcl}

\usepackage{amsfonts}
\usepackage{amsmath}
\usepackage{amsthm}
\usepackage{amsbsy} 
\usepackage{amssymb}
\usepackage{enumerate}
\usepackage{algorithm, algorithmic}
\usepackage{tikz}
\usepackage{amsfonts}

\usepackage{epstopdf}

\usetikzlibrary{shapes}

\newtheorem{theorem}{Theorem}
\newtheorem{lemma}{Lemma}

\newtheorem{corollary}{Corollary}
\newtheorem{remark}{Remark}

\newcommand{\x}{{\bf x}}
\newcommand{\y}{{\bf y}}
\newcommand{\z}{{\bf z}}
\newcommand{\w}{{\bf w}}
\usepackage{color}

\begin{document}

\title{Detecting the large entries of a sparse covariance matrix in sub-quadratic time}

\author{{\sc Ofer Shwartz} \\[2pt]
{\sc and}\\[6pt]
{\sc Boaz Nadler} \\[2pt]
Department of Computer Science\\
Weizmann Institute of Science\\
Rehovot, Israel}

\maketitle

\begin{abstract}
{The covariance matrix of a $p$-dimensional random variable is a fundamental quantity in data analysis. Given \(n\) i.i.d. observations, it is typically estimated by the sample covariance matrix, at a computational cost of \(O(np^{2})\) operations. When \(n,p\) are large, this computation may be prohibitively slow. Moreover, in several contemporary applications, the population matrix is approximately sparse, and only its few large entries are of interest. This raises the following question, at the focus of our work:\ Assuming approximate sparsity of the covariance matrix, can its large entries be detected much faster, say in sub-quadratic time, without explicitly computing all its \(p^{2}\) entries?\
In this paper, we present and theoretically analyze two randomized algorithms that detect the large entries of an approximately sparse sample covariance matrix using only $O(np\text{ poly log } p)$ operations.
Furthermore, assuming sparsity of the population matrix, we derive sufficient conditions on the underlying random variable and on the number of samples \(n\), for the sample covariance matrix to satisfy our approximate sparsity requirements.
Finally, we illustrate the performance  of our algorithms
via several simulations.}
{sparse covariance matrix, sub-quadratic time complexity, multi-scale group testing.}
\end{abstract}

\section{Introduction}
Let $Z=(Z_1,\ldots,Z_p)$ be a \(p\)-dimensional real valued random variable with a \(p\times p\) covariance matrix $\Sigma$. Given $n$ i.i.d. samples of $Z$, denoted $\{{\bf z}_i\}_{i=1}^n$, a fundamental task in statistical inference is to estimate $\Sigma$. A standard estimator of $\Sigma$ is the sample covariance matrix \(S=\frac{1}{n-1}\sum_i ({\bf z}_i-\bar{\bf z})({\bf z}_i-\bar{\bf z})^T\), where $\bar{\bf z}$ is the sample mean. Direct computation of all $p^2$ entries of $S$ requires $O(np^2)$ operations. In various contemporary data analysis applications, where both \(p\) and $n$ are large, this computation may be prohibitively slow and challenging in terms of memory and storage.

In several applications, however, the population covariance matrix is approximately sparse, whereby only its few large entries are of interest, and the remaining entries are either small or even precisely equal to zero. Applications leading to sparse covariance matrices include, among others, gene arrays and biological networks \cite{butte_2000}, social networks, climate data and f-MRI scans.
 As we are only interested in the large entries, a key question is whether these can be computed significantly faster, possibly by cleverly detecting their locations, without directly computing the entire matrix.

In this paper we present and theoretically analyze two different randomized algorithms with sub-quadratic time complexity, to detect and subsequently compute the large entries of a sparse sample covariance matrix. First, in Section \ref{sec:sfft_based_algo}, we present a reduction of this problem to the sparse-Fast-Fourier-Transform (sFFT) \cite{Nearly_Optimal_sFFT}.
A solution to our task then directly follows by invoking multiple calls to the recently developed randomized sFFT sub-linear time algorithm.
Next, in Section \ref{sec:tree_algo} we present a simpler and more direct algorithm, based on the construction of $O(\log p)$ binary random trees. We prove that under suitable assumptions on the sparsity of the matrix \(S\), both algorithms are guaranteed, with high probability,  to locate its large entries. Furthermore, their runtimes are  $O(nrp\log^3p)$ operations for the sFFT-based algorithm, and \(O(nrp\log^2 p)\) operations for the tree-based method, where $r$ is a bound on the number of large entries in each row of \(S\). By suitable normalization of the input data, both algorithms can also detect large entries of a sparse sample \textit{correlation matrix}, whose entries are  $\frac{S_{ij}}{\sqrt{S_{ii}S_{jj}}}$.

The theoretical analysis of the two algorithms of Sections \ref{sec:sfft_based_algo} and \ref{sec:tree_algo} relies on the assumption that $S$ is approximately sparse. In reality, in various applications one may only assume that the population matrix is approximately sparse. To this end, in Section \ref{sec:num_samples} we provide sufficient conditions on $\Sigma$, the underlying random variable \(Z\) and the number of samples $n$ that ensure, w.h.p., the approximate sparsity of $S$.

Finally, in Section \ref{sec:simulations} we empirically compare the tree-based algorithm with   other methods to detect large entries of $S$. In addition, we illustrate on artificially generated data that includes near duplicates the potential applicability of our algorithm for the \textit{Near Duplicate Detection} problem, common in document
corpora analysis \cite{xiao_2011_PPJOIN}.

\subsection{Related works}

In the statistics literature, the problem of sparse covariance estimation has received significant attention, see for example \cite{Bickel_Levina_2008,Cai_Liu_2011,Karoui_2008,Bien_2011,Chaudhur_07}. As discussed in Sections \ref{sec:sfft_based_algo} and \ref{sec:tree_algo} below, under suitable sparsity assumptions, our algorithms are guaranteed to find with high probability all entries of the sample covariance matrix which are larger, in absolute value, than some threshold \(\mu\). The resulting matrix is then nothing but the sample covariance matrix, hard-thresholded at the value \(\mu\).
The statistical properties of such thresholding were intensively studied, see for example  \cite{Bickel_Levina_2008,Karoui_2008,Cai_Liu_2011} and references therein. The main motivation of these works was to derive a more accurate estimate of the population covariance matrix \(\Sigma\), assuming it is sparse, thus overcoming the asymptotic inconsistency of $S$ in the operator norm, in the joint limit \(p,n\to\infty\) with $p/n\to c$, see for example  \cite{karoui_2008_spectrum}. Moreover, hard-thresholding with a threshold that slowly tends to zero as \(p,n\to\infty\), was proven to be asymptotically minimax optimal under various sparsity models. Our work, in contrast, is concerned with the {\em computational effort} of computing this thresholded estimator, mainly for finite \(p,n\) and relatively large thresholds.

Focusing on the computational aspects, first of all note that the sample covariance matrix $S$ can be represented as a product of two suitable matrices. Hence, using fast matrix multiplication methods, all its entries can be computed faster than \(O(np^{2})\). Currently, the fastest matrix multiplication algorithm for square matrices of size $N\times N$ has a  time-complexity of $O(N^{2.3727})$   \cite{williams_2012_multiplying}. Hence, by expanding (with zero padding)\ the input sample matrix to a square matrix, all entries of $S$ can be exactly evaluated using $O(\max\{n,p\}^{2.3727})$ operations.

 In recent years, several works developed fast methods to \textit{approximate} the product of two matrices, see for example \cite{drineas_2006}, as well as \cite{iwen_spencer_2009} and \cite{pagh_2013} which assume that the product is approximately sparse. In particular, in a setting similar to the one considered in this paper, the method of \cite{pagh_2013}, based on fast approximate matrix multiplication, can detect the large entries of a sparse covariance matrix in time complexity comparable to ours. 

In addition, since the entries of $S$ can be represented as inner products of all-pairs vectors, our problem is directly related to the \textit{Maximum Inner Product Search} (MIPS) problem \cite{ram_2012, shrivastava_2014}. In the MIPS\ problem, given a large dataset \(\{{\bf x}_i\}\),  the goal is to quickly find, for any query vector \({\bf y}\), its maximal inner product $\max_i \langle {\bf y},{\bf x}_i\rangle$. \cite{ram_2012} presented three algorithms to retrieve the maximal inner product, which can be generalized to find the  $k$ largest values. While \cite{ram_2012} did not provide a theoretical analysis of the runtime of their algorithms, empirically on several datasets they were significantly faster than direct computation of all inner products.  Recently, \cite{shrivastava_2014} presented a simple reduction from the \textit{approximate}-MIPS problem, of finding $\max_i \langle {\bf y},{\bf x}_i\rangle$ up to a small distortion \(\epsilon\), to the well-studied $k$ \textit{nearest neighbour}  problem ($k$NN).
This allows to solve the approximate-MIPS problem using any  $k$NN procedure. The   (exact or approximate) $k$NN search problem has been extensively studied in the last decades, with several fast algorithms, see \cite{shakhnarovich_book_2006,osipov_rokhlin_2013,arya_1998} and references therein. For example, $kd$-tree \cite{kd_tree} is a popular exact algorithm for low dimensional data, whereas \textit{Local-sensitive-hashing} (LSH) approximation methods \cite{datar_indyk_2004} are more suitable to high dimensions. Combining the reduction of \cite{shrivastava_2014} with  the  LSH-based algorithm of \cite{peled_indyk_2012} yields an approximate solution of the MIPS problem  with $O(np^\gamma\text{ poly} \log p)$ query time, where $\gamma \in (0,1)$ controls the quality of the approximation.  These
methods can be exploited to approximate the sample covariance matrix, perhaps with no assumptions on the matrix
but with slower (or no) runtime guarantees. In section \ref{sec:simulations}, we empirically compare our sub-quadratic tree-based algorithm to some of the above methods.

Another line of work related to fast estimation of a sparse covariance matrix is the matrix sketching problem  \cite{nowak_2013_sketching}. Here, the goal is to recover an unknown matrix $X$, given only partial knowledge of it, in the form of a linear projection $AXB$ with known matrices $A$ and $B$. The matrix sketching problem is more challenging, since we are given only partial access to the input. In fact, recent solutions to this problem \cite{adler_2013} have time complexity at least $O(np^2)$.

A more closely related problem is the \textit{all-pairs similarity search} with respect to the cosine similarity or the Pearson-correlation. Here, given a set of vectors \(\{\x_i\}\), the task is to find the pair with highest cosine similarity, \(\max_{i\neq j}{\x_i}^t{\x_j}/\|\x_i\| \|\x_j\|\). Two popular exact methods, suitable mainly for sparse inputs, are \textit{All-Pairs} \cite{Bayardo_2007} and \textit{PPJoin+} \cite{xiao_2011_PPJOIN}. As in the nearest-neighbours search problem,  hashing can be adapted to obtain various approximation algorithms \cite{charikar_2002}.
These algorithms can be used to rapidly compute a sparse correlation matrix. In a special case of the all-pairs similarity search, known as  the \textit{Light Bulb Problem} \cite{valiant_1988}, one is given \(p\) boolean vectors all of length \(n\) and taking values $\pm 1$. The assumption is that all vectors are uniformly distributed on the \(n\)-dimensional boolean hypercube, apart from a pair of two vectors which have a Pearson-correlation $\rho\gg\sqrt{\log p/n}$. The question is how fast can one detect this pair of correlated vectors. LSH type methods as well as bucketing coding \cite{dubiner_2010}
solve this problem with a sub-quadratic time complexity of $O(np^{g(\rho)})$, for a suitable function $g(\rho)$ of the correlation coefficient. 
More recently, \cite{valiant_2015} developed a method with expected sub-quadratic time complexity of $O(np^{1.62})$ operations, where the exponent value 1.62 is independent of $\rho$ and directly related to the complexity of the fastest known matrix multiplication algorithm. 
It is easy to show that the methods proposed in our paper can solve the Light Bulb Problem using only $O(np\text{poly}\log p)$ operations,
but assuming a significantly stronger correlation of  $\rho > O(\sqrt{\frac{p}{n}})$.
Furthermore, our methods offer an interesting tradeoff between time complexity and correlation strength \(\rho\). For example, our tree-based algorithm with \(O(p^\alpha)\) trees (instead of \(O(\log p)\)), can detect weaker correlations of strength \(\rho>O(\sqrt{\frac{p^{1-\alpha}}{n}})\) with a time complexity of \(O(np^{1+\alpha}\text{ poly}\log p)\) operations.

\section{Notation and Problem Setup}
\label{sec:notation}
For simplicity, in this paper we assume the input data $(\z_1, \dots, \z_n)$ is real-valued, though the proposed methods can be easily modified to the complex valued case.
For a vector $\x\in\mathbb{R}^n$, we denote its $i$-th entry by $\x_i$ (or $(\x)_i$), its $L_2$-norm by $||\x|| := \sqrt{\sum_{i=1}^n\x_i^2}$, and its $L_0$-norm by $||\x||_0:=| \{ i:\x_i \neq 0\} | $.
 Similarly, for a matrix $A\in\mathbb{R}^{p\times p}$ we denote its $k$-th row by $A_k$, its $(i,j)$-th entry by $A_{ij}$ (or $(A_{i})_{j}$), its $L_2$-norm by $||A||: = \sup_{\x \neq 0}\frac{||A\x||}{||\x||}$, and its Frobenius norm by $||A||_F := \sqrt{\sum_{ij}A_{ij}^2}$. For an integer $a\in\mathbb{N}$, let
$ [a]:=\{1, 2, \dots, a\}$.
To simplify notations, the inner product between two complex vectors $\x,\y\in \mathbb{C}^n$  is defined with a normalization factor,
$$ \langle \x, \y\rangle = \frac{1}{n-1}\sum_{i=1}^n \x_i \y_i ^\dagger$$
where $\dagger$ represents the complex conjugate, and will be relevant only when we discuss the Fourier transform which involves complex valued numbers.

Given the input data $\{\z_1,\dots, \z_n\}$, let $\x_i=((\z_{1})_i - (\bar \z_i),\cdots,(\z_{n})_i - (\bar \z_{n}) )\in\mathbb{R}^n$ be a mean centered vector of observations for the $i$-th coordinate of $Z$. This leads to the simple representation of $S$, in terms of inner products
\begin{equation}
\label{eq:s_rep_ip}
S_{ij}=\langle \x_i, \x_j\rangle, \qquad 1\leq i,j\leq p.
\end{equation}

For a matrix $A\in\mathbb{R}^{p\times p}$ and a threshold parameter $\mu$, we define the set of large entries at level $\mu$ in the $k$-th row to be$$J_\mu(A_{k}) = \{j \ :\ |A_{kj}|\geq\mu\}$$
and the set of all its large entries
at level $\mu$ by
$$ J_\mu(A) = \bigcup_{k=1}^p \{k\}\times J_\mu(A_k).$$
As for the definition of matrix sparsity, we say that a matrix $A$ is $(r,\mu)$\textit{-sparse} if for every row $k$, $|J_{\mu}(A_{k})| \leq r$. We say that $A$ is $(r, \mu, R,q)$-sparse if it is $(r,\mu)$-sparse and for every $k\in[p]$ the remaining small entries in the $k$-th row are contained in a $L_q$-ball of radius $R$,
$$\left(\sum_{j\not\in J_\mu(A_k)}|A_{kj}|^q\right)^{1/q}\leq R.$$
Here we only consider the case where $q=2$ and $R <\mu/2$.

\paragraph{Problem Setup.} Let $\{\z_1,\dots, \z_n\}$ be $n$ input vectors whose covariance matrix \(S\) is $(r,\mu)\)-sparse, with $r\ll p$. In this paper we consider the following task: Given $\{\z_i\}_{i=1}^n$ and the threshold value \(\mu\), find the set  $J_\mu(S)$, which consists of all entries of $S$ which are larger than $\mu$ in absolute value.
A naive approach is to explicitly compute all  $p^2$ entries of \(S\) and then threshold at level $\mu$. This requires $O(np^2)$ operations and may be quite slow when $p \gg1$. In contrast, if an oracle gave us the precise
locations
of all large entries, computing them directly would require only \(O(npr)\) operations.

The key question studied in this paper is whether we can find the set  $J_\mu(S)$, and compute the corresponding entries significantly faster than \(O(np^{2})\). In what follows, we present and analyze two sub-quadratic algorithms to discover $J_\mu(S)$, under the assumption that the matrix $S$ is approximately sparse.

\section{Sparse covariance estimation via sFFT}
\label{sec:sfft_based_algo}
The first solution we present is based on a reduction of our problem to that of  multiple independent instances of sparse Fourier transform calculations. Whereas standard fast-Fourier-transform (FFT) of a vector of length \(p\) requires \(O(p\log p)\) operations, if the result is a-priori known to be approximately sparse, it can be computed in sub-linear time. This problem, known as sparse-FFT (sFFT), has been intensively studied in the past few years, see \cite{Nearly_Optimal_sFFT, akavia_2010, gilbert_2002, iwen_2010}. As described below, the computational gains for our problem then directly follow by an application of one of the available sFFT algorithms.

In what follows we present this reduction, and focusing on the algorithm of  \cite{Nearly_Optimal_sFFT} we analyze under which sparsity conditions on \(S\) it is guaranteed to succeed w.h.p. To this end, we use the following definitions for the discrete Fourier transform (DFT) $\mathcal{F}:\mathbb{C}^p\rightarrow\mathbb{C}^p$
and its inverse $\mathcal{F}^{-1}$,
$$
(\mathcal{F}[\x])_j=\sum_{l=1}^{p}\x_l\omega^{-jl}\qquad
(\mathcal{F}^{-1}[\x])_j=\frac{1}{p}\sum_{l=1}^{p}\x_l\omega^{jl}
$$
where $\omega=\exp(2\pi {\bf i}/p)$, \({\bf i}=\sqrt{-1} \). Without any assumptions on the input $\x$, the fastest known method to compute its DFT is the FFT which requires $O(p\log p)$ operations.
If the vector \(\mathcal F[\x]\) is approximately $r$-sparse, with only $r$ large entries, it is possible to estimate it in \textit{sub-linear time}, by detecting and approximately evaluating only its large entries. In particular, \cite{Nearly_Optimal_sFFT} developed a randomized algorithm to compute the sparse Fourier transform with time complexity of $O(r\log p\log(p/r))$, which we refer to here as the sFFT algorithm.
Formally, for any input vector $\x\in\mathbb{C}^p$ and parameters $\alpha, \delta,r$, the sFFT algorithm returns a sparse vector $\hat \x$ such that, w.p. $\geq 2/3$
\begin{equation}
\label{eq:sfft_algo_res}
||\mathcal{F}[ \x]-\hat \x||\leq (1+\alpha)\min_{||\y||_0\leq r}||\mathcal{F}[\x]-\y|| +\delta||\mathcal{F}[\x]||. \end{equation}
The algorithm time complexity is $O(\frac{r}{\alpha}\log(p/r)\log(p/\delta))$, though for our purposes we assume that $\alpha$ is fixed (e.g. $\alpha=1$) and hence ignore the dependence on it.
The output \(\hat\x\) of the sFFT algorithm is represented as a pair \((J,\y)\) where \(J\subset\{1,\ldots,p\}\) is a set of indices and $\y\in\mathbb{C}^{|J|}$ are the values of $\hat\x$ at these indices, namely $\hat\x|_J=\y$ and \(\hat\x|_{J^c}=0.\)

We now present the reduction from the sparse covariance estimation problem to sparse FFT. To this end,
let us define a matrix $W = ({\bf w}_1, \dots, \w_p)\in \mathbb{C}^{n\times p}$ whose $j$-th column is $\w_j=\frac{1}{p}\sum_{l=1}^{p}\x_l\, \omega^{-jl}$. Note that using standard FFT methods, $W$ can be calculated in $O(np\log p)$ operations.
The following lemma describes the relation between the matrix $S$ and the matrix $W$.

\begin{lemma}\label{lemma:sfft_reduction}
For any $k\in[p]$, let ${\bf u}_k=(\langle \x_k, \w_{1}\rangle,\dots,\langle\x_k, \w_{p}\rangle)\in\mathbb{C}^p$. Then,
\begin{equation}
\label{eq:sfft_red}
\mathcal{F}[{\bf u}_k]=S_{k}.
\end{equation}
\end{lemma}
\begin{proof}
Eq. (\ref{eq:sfft_red}) follows directly by applying the inverse DFT on $S_k$,
$$ (\mathcal{F}^{-1}[S_k])_j = \frac{1}{p}\sum_{l=1}^{p}\langle \x_k, \x_l\rangle\omega^{jl}=\langle\x_k,\frac{1}{p}\sum_{l=1}^{p} \x_l\omega^{-jl} \rangle=\langle \x_k, \w_j\rangle=({\bf u}_k)_j.$$

\end{proof}

According to Lemma \ref{lemma:sfft_reduction}, each row $S_k$ of $S$ is the DFT of an appropriate vector \({\bf u}_k\). Since by assumption $S_k$ is approximately sparse, we may thus find its set of large indices  in sublinear-time by applying the sFFT algorithm on the input \({\bf u}_k\).
We then explicitly compute the corresponding entries in \(S_{k}\) directly from the original data \(\x_{1},\ldots,\x_p\).

Computing all \(p\) entries of all \(p\) vectors \(\{{\bf u}_k\}\) requires a total of \(O(np^{2})\) operations. However, with this time complexity we could have computed all entries of the matrix \(S\) to begin with. The key point that makes this reduction applicable is that the sFFT algorithm is sub-linear in time and to compute its output it reads at most $O(r\log(p/\delta)\log(p/r))$ coordinates of ${\bf u}_k$. In particular, it does not require a-priori evaluation of all \(p\) entries of each vector \({\bf u}_k\). Hence, we may compute on-demand only the entries of ${\bf u}_k$ requested by the algorithm. Since computing a single entry of ${\bf u}_k$ can be done in $O(n)$ operations, the total number of operations required for a single sFFT\ run is \(O(nr\log(p/\delta)\log(p/r))\).

To detect the large entries in all rows of \(S\), all \(p\) outputs of the sFFT algorithm should simultaneously satisfy Eq. (\ref{eq:sfft_algo_res}) with high probability.  Given that each sFFT run succeeds w.p. $\geq 2/3$, we show below that it suffices to invoke  $m=O(\log p)$ independent queries of the sFFT algorithm on each input ${\bf u}_k$. The output for each row is the union of the large indices found by all sFFT runs.

In more detail, for each row $k$ let $(J_{1},\y_1), \dots, (J_m,\y_m)$ be the outputs (indices and values)\ of sFFT on $m$ independent runs with the same input ${\bf u}_k$. Our approximation for the \(k\)-th row of \(S\) is then
\begin{equation}
\tilde S_{kj}=
\left\{
\begin{array}{cl}
S_{kj} & j\in \bigcup_i J_i \\
0      & \mbox{otherwise}
\end{array}
\right.
\end{equation}
The following lemma and corollary prove that $m=O(\log p)$ runs suffice to detect all large entries of $S$ with a constant success probability.

\begin{lemma}\label{lemma:sfft_prob} Let $m= \lceil\log(3p)/\log(3)\rceil=O(\log(p))$. Then, for each $k\in[p]$ the following inequality holds w.p. $\geq1- \frac{1}{3p}$,
$$ ||S_{k} -\tilde S_{k } ||\leq (1+\alpha)\min_{|| \y||_0\leq r}||S_{k}-\y||  +\delta||S_k||. $$
\end{lemma}

\begin{corollary}\label{cor:sfft_cor}Let $m= \lceil\log(3p)/\log(3)\rceil$ as in the previous lemma. Then with probability $\geq 2/3$, simultaneously for all rows $k\in[p]$
\begin{equation}
        \label{eq:sfft_cor}
||S_{k} -\tilde S_{k } ||\leq (1+\alpha)\min_{\y:||\y||_0\leq r}||S_{k}-\y||+\delta||S_k||.
\end{equation}
\end{corollary}
{The proof of corollary \ref{cor:sfft_cor} follows from a standard union-bound argument, details omitted.
}

To conclude, we use multiple runs of the sFFT algorithm to find a set \(I\) of indices in $S$, whose entries \(S_{ij}\) may be potentially large, and which we then evaluate explicitly. The resulting approximation of $S$ is compactly represented as   $\{((i,j), S_{ij})\}_{(i,j)\in I}$.
This procedure is summarized in Algorithm \ref{algo:sfft_algo}. The following Theorem, proven in the appendix, provides a bound on its runtime and a guarantee for the  accuracy of its output.

\begin{algorithm}[t]
\caption{\tt sFFTCovEstimation($\x_1, \dots, \x_p, r, R, \epsilon$)}
\label{algo:sfft_algo}
\begin{algorithmic}[1]
\REQUIRE \ \\$(\x_1, \dots, \x_p)$: $p$ vectors of dimension $n$. \\$r$: bound on the number of large entries in each row.\\ $R, \epsilon$: sparsity parameters.
\ENSURE $\tilde S$: Compact representation of the large entries of $S$.
\STATE
Compute $W = (\w_1, \dots, \w_p)\in \mathbb{C}^{n\times p}$ using FFT, where $\w_j=\frac{1}{p}\sum_{l=1}^{p}\x_l\omega^{-jl}$
\STATE Set $I=\emptyset$
\STATE Compute all diagonal entries \(S_{ii}\) and the value  $M=\max_{i}{S_{ii}}$
\FOR {$j=1, \dots, p$}
\STATE
Calculate $(J_{1},\y_1), \dots, (J_m,\y_m)$ by running  $m=O(\log p)$ sFFT queries with input ${\bf u}_k=(\langle \x_k, \w_{1}\rangle,\dots,\langle \x_k, \w_{p}\rangle), \delta = \frac{\epsilon}{R+\sqrt{r}M}$ and $\alpha = 1$
\STATE Add $\{ k\}\times(\bigcup_{i=1}^{m} J_i)$ to $I$
\ENDFOR
\FOR {$(i,j)\in I$}
\STATE Calculate $S_{ij}=\langle \x_i, \x_j\rangle$
\ENDFOR
\RETURN $\tilde S = \{((i,j), S_{ij})\}_{(i,j)\in I}$\end{algorithmic}
\end{algorithm}

\begin{theorem}\label{claim:sfft_algo_correctness}
Assume $S$ is $(r, \mu,R,2)$-sparse, where $\mu > 2R+\epsilon$ for known $R, \epsilon>0$, and let \(M=\max_i S_{ii}\).
Then, Algorithm \ref{algo:sfft_algo}, which invokes the sFFT algorithm with parameters $\delta = \frac{\epsilon}{R+\sqrt{r}M}$ and $\alpha =1$, has a runtime of $O(nrp\log^2p\log((R+\sqrt{r}M)p/\epsilon))$ operations, and w.p. $\geq$ 2/3, its output set $I$ is guaranteed to include the set
$J_\mu(S)$ of all large entries of $S$.
\end{theorem}

\section{Tree-based Algorithm}
\label{sec:tree_algo}
We now present a second more direct method to efficiently detect and compute the large entries of a sparse covariance matrix. This method, which assumes the threshold $\mu$ is a priori known,  is based on a bottom-up construction of $m$ random binary trees. To construct the $l$-th tree, we first place at  its $p$ leaves the following $n$ dimensional vectors $\{\eta_{lj}\x_j\}_{j=1}^p$, where $\eta_{lj}\overset{i.i.d.}{\sim}N(0,1)$. Then, at higher levels, the $n$-dimensional vector in each parent node is the sum of its offsprings. After the construction of the trees, the main idea of the algorithm is to make recursive coarse-to-fine statistical group tests, where all indices of various subsets $\mathcal{A}\subseteq[p]$ are simultaneously tested whether they contain at least one large entry of $S$ or not.

In more details, given as input a row  $k\in[p]$ and a set $\mathcal{A}\subseteq[p]$, we consider the following query or hypothesis testing problem:
$$ \mathcal{H}_0:J_\mu(S_k)\cap \mathcal{A}=\emptyset \quad \text{vs.} \quad
\mathcal{H}_1:J_\mu(S_k)\cap \mathcal{A}\neq\emptyset.
$$
Assuming $S$ is $(r, \mu, R, 2)$-sparse, with $R<\mu/2$, one way to resolve this query is by computing the following quantity
 \begin{equation}
\label{eq:tree_alg_sum}
F(k, \mathcal{A})=\sum_{j\in \mathcal{A}}\langle \x_{k}, \x_{j}\rangle^2.
\end{equation}
Indeed, under $\mathcal{H}_0$, $F(k,\mathcal{A})\leq R^2$, whereas under $\mathcal{H}_1$, $F(k, \mathcal{A})\geq \mu^2$.

A direct calculation of (\ref{eq:tree_alg_sum})
requires $O(n|\mathcal{A}|)$ operations, which may reach up to  $O(n p)$ when $|\mathcal{A}|$ is large. Instead, the algorithm {\em estimates} the value of $F(k, \mathcal{A})$ in Eq. (\ref{eq:tree_alg_sum}) using  $m$  i.i.d. samples $\{ y_l\}_{l=1}^{m}$ of the random variable
\begin{equation}
\label{eq:tree_alg_sum_rv}
Y(\mathcal{A}) = \langle \x_k,\sum_{j\in \mathcal{A}}\eta_j\x_j\rangle
\end{equation}
where $(\eta_1, \dots, \eta_p)\sim N(0,I_p)$. Since by definition $\mathbb{E}Y(\mathcal{A})=0$ and $Var[Y(\mathcal{A})]=F(k, \mathcal{A})$,  Eq. (\ref{eq:tree_alg_sum}) can be estimated by the empirical variance of the $m$ variables $\{y_l\}_{l=1}^m$. Again,  given a specific realization of $\eta_j$'s, direct calculation of Eq. (\ref{eq:tree_alg_sum_rv}) also requires $O(n|\mathcal{A}|)$ operations. Here our construction of $m$ binary trees comes into play,  as it efficiently provides us samples of  $\sum_{j\in \mathcal{A}}\x_j\eta_j$, for any subset of the form $\mathcal{A}=\{2^hi, \dots, 2^h(i+1) \}$, where $h\in[\log_{2} p]$ and $i\in[p/2^h] $. As described in section \ref{sec:quary} below, after the pre-processing stage of constructing the \(m\) trees, each query requires only $O(nm)$ operations. Moreover, we show that $m=O(\log p)$ trees suffice for our purpose, leading to $O(n\log p)$ query time.

To find large entries in the $k$-th row, a divide and conquer method is applied: start with $\mathcal{A}=\{1, \dots, p\}$, and check if $F(k, \mathcal{A})$ of Eq. (\ref{eq:tree_alg_sum}) is large by invoking an appropriate query. If so, divide $\mathcal{A}$ into two disjoint sets and continue recursively. With this approach we reduce the number of required operations for each row from $O(np)$ to $O(nr\log^2 p)$, leading to an overall time complexity of $O(nrp\log^2p)$.

\subsection{Preprocessing Stage - Constructing the trees}

To efficiently construct the $m$ trees, first $mp$ i.i.d. samples $\{\eta_{lj}\}$ are generated from a Gaussian distribution $N(0,1)$. For simplicity, we assume that $p=2^L$, for some integer $L$, leading to a full binary tree of height $L+1$.
Starting at the bottom, the value at  the $j$-th leaf in the $l$-th tree is set to be $n$-dimensional vector\ $\eta_{lj}\x_j$. Then, in a bottom-up construction, the vector of each node is the sum of its two offsprings. Since, each  tree has $2p-1$ nodes, and calculating the vector of each node requires $O(n)$ operations, the construction of $m$ i.i.d. random trees requires $O(nmp)$ operations.

 For future use, we introduce the following notation: for a given tree $T$, we denote by $T(h,i)$ the $i$-th node at the $h$-th level, where the root is considered to be at level zero. Furthermore,  we denote by $I(h,i)$ the set of indices corresponding to the leaves of the subtree starting from  $T(h,i)$.
The vector stored at the node $T(h,i)$ is the following $n$-dimensional random variable
$$ Val(T(h,i)) = \sum_{j\in I(h,i)}\eta_j\x_j$$
whose randomness comes only from the random variables $\eta_j$ (we here consider the samples $\{ \x_j\}$ as fixed). The entire tree construction procedure is described in Algorithm \ref{algo:tree_preprocess}.

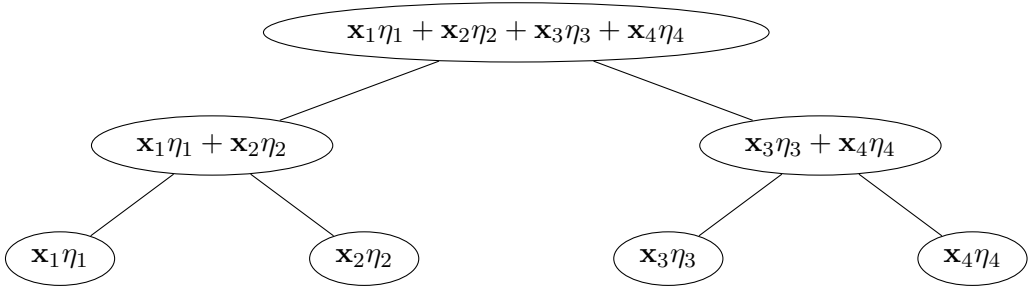
\begin{figure}[t!]
\label{fig:tree_example}
\centering
\begin{tikzpicture}[auto,
    level 1/.style={sibling distance=80mm},
    level 2/.style={sibling distance=40mm}]
\node [ellipse,draw] (n1){$\x_1\eta_{1}+\x_2\eta_{2}+\x_3\eta_{3}+\x_4\eta_{4}$} [sibling distance=60mm]
        child {
            node[ellipse,draw] (n2) {$\x_1\eta_{1}+\x_2\eta_{2}$}
            child {
                node[ellipse,draw] (n4) {$\x_1\eta_{1}$}
            }
            child {
                node[ellipse,draw] (n5) {$\x_2\eta_{2}$}
            }
        }
        child {
            node[ellipse,draw] (n3) {$\x_3\eta_{3}+\x_4\eta_{4}$}
            child {
                node[ellipse,draw] (n6) {$\x_3\eta_{3}$}
            }
            child {
                node[ellipse,draw] (n7) {$\x_4\eta_{4}$}
            }
        }
        ;
\end{tikzpicture}
\caption{Illustration of a random tree, for $p=4.$}
\end{figure}

\begin{algorithm}
\caption{\tt ConstructTrees$(\x_1, \dots, \x_p, m)$}
\label{algo:tree_preprocess}
\begin{algorithmic}[1]
\REQUIRE \ \\
$(\x_1, \dots, \x_p)$: $p$ vectors of dimension $n$. \\
$m$: number of trees.
\ENSURE $m$ random binary trees $(T_1, \dots, T_m).$
\STATE Generate $\eta_{lj}$, for $j\in[p]$ and $l\in[m]$, where $\eta_{lj}\overset{i.i.d.}{\sim}N(0,1)$
\FOR {$l=1,\dots,m$}
\FOR {$j=1,\dots,p$}
\STATE Set $Val(T_l(L, j)) = \eta_{lj} \x_j$
\ENDFOR
\FOR {$h=L-1,\dots, 0$}
\FOR {$i=1,\dots, 2^{h}$}
\STATE Set $Val(T_l(h, i))=Val(T_l(h+1, 2i))+Val(T_l(h+1, 2i+1))$
\ENDFOR
\ENDFOR
\ENDFOR
\end{algorithmic}
\end{algorithm}

\subsection{Algorithm description}
\label{sec:quary}
As mentioned before, we assume that the threshold $\mu$ is an input parameter to the algorithm. Given a set $\mathcal{A}=I(h,i)$ and a vector $\x_k$, to estimate (\ref{eq:tree_alg_sum}) the algorithm uses $m$ i.i.d. samples of (\ref{eq:tree_alg_sum_rv}) obtained from the relevant node $T(h,i)$ of the tree. For this, let $y_1,
\dots, y_m$ be the $m$ i.i.d. samples of $Y=Y(I(h,i))$,
$$ y_l = \langle\x_k, Val(T_l(h,i))\rangle=\sum_{j\in I(h,i)} \langle\x_{k}, \x_{j}\rangle\eta_{lj}.$$
Since $\mathbb{E}Y=0$ and $$\sigma^2:=Var(Y) =\sum_{j\in I(h,i)} \langle \x_{k}, \x_{j}\rangle^{2}$$
a natural estimator for (\ref{eq:tree_alg_sum}) is the sample variance
$${\hat\sigma}^2:=\frac{1}{m}\sum_{l=1}^my_l^2.$$
Recall that for a matrix $S$ which is $(r, \mu, R,2)$-sparse with \(R<\mu/2\), the exact value\ of Eq. (\ref{eq:tree_alg_sum}) allows us to perfectly distinguish between the following two hypotheses
\begin{equation}
\label{eq:tree_hyp}
\mathcal{H}_0: I(h,i)\cap J_\mu(S_{k})= \emptyset \quad\text{vs.}\quad
 \mathcal{H}_1: I(h,i)\cap J_\mu(S_{k})\neq \emptyset. \end{equation}
Given the estimate $\hat \sigma^2$, the algorithm considers the following test
\begin{equation}
\label{eq:tree_test}
{\hat\sigma}^2\mathop{\gtrless}_{\mathcal{H}_0}^{\mathcal{H}_1} \frac{3\mu^2}{4}
\end{equation}
If $\hat\sigma^2 \geq \frac{3\mu^2}{4}$, the algorithm continues recursively in the corresponding subtree. Lemma \ref{lemma:tree_test} below shows that $m=O(\log \frac{1}{\delta})$ trees suffice to correctly distinguish between the two hypotheses w.p. at least $1-\delta$.

In summary, the algorithm detects large entries in the $k$-th row by processing  the \(m\) random trees with the vector \(\x_k\); starting from the root, it checks if one of its children contains a large entry using the query presented in Eq. (\ref{eq:tree_test}), and continues recursively as required. This recursive procedure is described in Algorithm \ref{algo:tree_rec}. Then, as described in Algorithm \ref{algo:tree_algo}, the complete algorithm to detect all large entries of a sparse covariance matrix applies Algorithm 3 separately to each row. As in Algorithm \ref{algo:sfft_algo}, the output of Algorithm \ref{algo:tree_algo} is a compact representation of the large entries of the matrix, as a set of indices $I\subset[p]\times[p]$ and their corresponding entries    $\{S_{ij}\}_{(i,j)\in I}$.

\begin{algorithm}[t]
\caption{{\tt Find}$(\x, h, i,m, \{T_l\},\mu)$}
\label{algo:tree_rec}
\begin{algorithmic}[1]
\REQUIRE \ \\ $\x$: input vector of dimension $n$.\\ $(h,i)$: tree-node index.
\\ $m$: number of trees.
\\ $\{ T_l\}$: a collection of $m$ trees. \\
$\mu$: threshold parameter.
\ENSURE The set of large entries in the current sub-tree.
\IF {$h= L$}
\RETURN $\{ i\}$
\ENDIF
\STATE Set $a =\frac{1}{m} \sum_{l=1}^m \langle \x, Val(T_{l}(h+1, 2i))\rangle^2$,
$b =\frac{1}{m} \sum_{l=1}^m \langle \x, Val(T_{l}(h+1, 2i+1))\rangle^2$
\STATE Set $\mathcal{S}_a=\emptyset, \mathcal{S}_b=\emptyset$\IF {$a \geq \frac{3\mu^2}{4}$}
\STATE $\mathcal{S}_a = {\tt Find}(\x, h+1, 2i,m, \{T_l\},\mu)$
\ENDIF
\IF {$b \geq \frac{3\mu^2}{4}$}
\STATE $\mathcal{S}_b = {\tt Find}(\x, h+1, 2i+1, m, \{T_l\},\mu)$
\ENDIF
\RETURN $\mathcal{S}_a \cup \mathcal{S}_b$
\end{algorithmic}
\end{algorithm}
\begin{algorithm}[t]
\caption{{\tt SparseCovTree}$(\x_1, \dots, \x_p, m, \mu)$}
\label{algo:tree_algo}
\begin{algorithmic}[1]
\REQUIRE \ \\
$(\x_1, \dots, \x_p)$: $p$ vectors of dimension $n$. \\
$m$: number of trees. \\
$\mu$: threshold parameter.
\ENSURE $\tilde S$: compact representation of large entries of $S$.
\STATE Construct \(m\) trees: $\{T_l\}={\tt ConstructTrees}(\x_1, \dots, \x_p,m)$
\STATE Set $I=\emptyset$
\FOR {$k=1,\dots, p$}
\STATE Add $\{k\} \times {\tt Find}(\x_k, 0, 1, m, \{T_l\},\mu)$ to $I$
\ENDFOR
\FOR {$(i,j)\in I$}
\STATE Calculate $S_{ij}=\langle \x_i, \x_j\rangle$
\ENDFOR
\RETURN $\{((i,j), S_{ij})\}_{(i,j)\in I}$
\end{algorithmic}
\end{algorithm}

\subsection{Theoretical analysis of the Tree-Based algorithm}\label{sec:algo_analysis}

Two key questions related to Algorithm \ref{algo:tree_algo} are: i) can it detect w.h.p. large entries of $S$; and ii) what is its runtime. In this section we study these questions under the assumption that the matrix $S$ is approximately sparse. We start our analysis with the following lemma, proven in the appendix, which shows that $m=O(\log \frac{1}{\delta})$ trees suffice for the test in Eq. (\ref{eq:tree_test}) to succeed w.p. at least $1-\delta$.

\begin{lemma}
\label{lemma:tree_test}
Assume $S$ is $(r, \mu, R, 2)$-sparse, where $R< \mu/2$. Then, for $m\geq64\log(\frac{1}{\delta})$,
\begin{align*} \Pr[\textbf{false alarm}]=\Pr\left[{\hat\sigma}^2\geq \frac{3\mu^2}{4}\;|\;\mathcal{H}_0\right] \leq  \delta\\
 \Pr[\textbf{misdetection}] =\Pr\left[{\hat\sigma}^2< \frac{3\mu^2}{4}\;|\;\mathcal{H}_1\right]  \leq\delta
 \end{align*}
\end{lemma}

\begin{remark}
While the constant of 64 is not sharp, the logarithmic dependence on $\delta$ is.
\end{remark}
\begin{remark}
The choice of $\eta$ to be Gaussian is rather arbitrary. Standard concentration inequalities \cite{vershynin_2010} imply that every sub-Gaussian distribution with zero mean and variance $1$ will yield similar results, albeit with different constants.
\end{remark}

 Next, assuming $S$ is approximately sparse, Theorem \ref{theorem:tree_prob} below shows that w.h.p., Algorithm \ref{algo:tree_algo} indeed succeeds in finding all large entries of $S$. Most importantly, its runtime is bounded by $O(nrp\log^{2}p)$ operations, both w.h.p. and in expectation.

\begin{theorem}
\label{theorem:tree_prob}
Assume $S$ is $(r, \mu, R, 2)$-sparse, where $R< \mu/2$. Let $I(m)$ be the set of indices returned by  Algorithm \ref{algo:tree_algo} with threshold $
\mu$ and $m$ trees. For a suitably chosen constant $C$ that is independent of $n$, $m, r$ and $p$, define
$$ f(p,m) = \Pr\left[ J_\mu(S) \subseteq I(m) \text{ and } \text{Runtime}\leq Cnpr\log ^{2}p\right].$$
Then, for $m(p) = \lceil64\log(2rpL^3) \rceil$, where $L=\log_2p + 1$,
\begin{enumerate}[(i)]
\item
For all $p\geq 8$, $f(p, m(p))\geq 2/3$.
\item
$f(p, m(p))\xrightarrow{p\rightarrow\infty}1. $
\item
$\mathbb{E}[\text{Runtime}]\leq C'nrp\log ^{2}p$, for some absolute constant $C'$.
\end{enumerate}
\end{theorem}
\begin{remark}The sparsity of $S$ is required only to bound the false alarm probability in Lemma \ref{lemma:tree_test}. If $S$ is not necessarily sparse, the algorithm will still locate w.h.p. all of its large entries. However, its runtime can increase significantly, up to $O(np^2\log^2p)$ operations in the worst case when $r=p$.
\end{remark}

\begin{remark} Note that since our tree-based algorithm analyzes each row of \(S\) separately, in fact $S$ need not be globally $(r, \mu, R, 2)$-sparse. Our algorithm is still applicable to a matrix \(S\) with $k$ non-sparse rows and all other rows $(r, \mu, R, 2)$-sparse. On the non-sparse rows our algorithm will be quite slow, and would analyze all the leaves of the tree. However, on the sparse rows, it would detect the large entries in time $O(npr\log^2 p)$. If $k=O(\log p)$ the overall
run time is, up to logarithmic factors in $p$, still $O(npr)$. 

\end{remark}
The runtime guarantees of Theorem \ref{algo:tree_algo} are probabilistic ones, where the algorithm runtime can reach up to $O(np^2\log^2p)$ operations, even when $S$ is sparse. For an a-priori known upper bound on the sparsity parameter $r$, one can slightly modify the algorithm to stop after $O(nrp\log^2p)$ operations. This modification may decrease the probability to detect all large entries, but still maintain it to be at least $2/3$. This is true since the event of finding all large entries contains the event of finding all large entries using a bounded number of operations, whereas the second event is not affected  by the modification.

\subsection{Comparison between the sFFT-based and tree-based algorithms  }
Table \ref{table:comparison} compares the main properties of the sFFT-based and the tree-based algorithms. Interestingly, even though the two algorithms are very different, their success relies on precisely the same definition of matrix sparsity. Moreover, the difference in the input parameters is not significant since knowledge of $R$ yields a lower bound for $\mu$, and vice versa.  Although the tree-based algorithm has a lower runtime complexity and is easier to understand and implement, the sFFT-based algorithm is still of interest as it illustrates a connection between two different problems which  seem unrelated at first glimpse; estimating the sparse Fourier transform of a given signal and detecting large entries of a sparse covariance matrix. Hence advances in sFFT  can translate to faster sparse covariance estimation.
\begin{center}
\begin{table}[H]
\small
\begin{tabular}{ |l | c |c |}
\hline
\textit{} & \textit{\textbf{sFFT-based }} & \textit{\textbf{SparseCovTree }} \\\hline
\textit{Sparsity assumptions} & $(r, \mu, R, 2)$-sparse with $\mu > 2R+\epsilon$ & $(r, \mu, R, 2)$-sparse with $\mu > 2R$ \\\hline
\textit{Runtime bound} & \multicolumn{1}{c|}{$O(nrp\log^2p\log((R+\sqrt{r}M)p/\epsilon)$} & \multicolumn{1}{c|}{$O(nrp\log^2p)$} \\\hline
\textit{Probability to detect all large entries}&   \multicolumn{1}{c|}{$\geq2/3$} & \multicolumn{1}{c|}{$\geq2/3$}\\\hline
\textit{Required input parameters }& $r, R, \epsilon$ and $M=\max\{S_{ii}\}$& $\mu$ only \\\hline
\textit{Dependencies on other algorithms}& Based on the sFFT algorithm & Standalone\\\hline
\end{tabular}
 \caption{Comparison between the sFFT-based and Tree-based algorithms.}
 \label{table:comparison}
\end{table}
\end{center}

\section{Relation between sample size, sparsity of 
$S$ and of $\Sigma$}
\label{sec:num_samples}
The theoretical analysis in the previous two sections assumed that the sample covariance matrix $S$ is approximately sparse. Typically, however, one may only assume that the population matrix $\Sigma$ is sparse, whereas $S$ is computed from a sample of $n$ observations $\z_1, \dots, \z_n$. In general, this may lead to a non-sparse matrix $S$, even when $\Sigma$ is sparse. Nonetheless, we show below that if the underlying r.v. $Z$ is sub-gaussian with a sparse covariance matrix  $\Sigma$, then given a sufficient number of samples  $n=O(p \log p)$, w.h.p. the corresponding sample covariance matrix $S$ is also approximately sparse, albeit with different parameters.

To this end,  recall that a random variable $Y$ is said to be sub-gaussian if
\begin{equation*}
||Y||_{\psi_2}:=\sup _{q\geq 1}q^{-1/2}(\mathbb{E}|Y|^q)^{1/q}<\infty
\end{equation*}
and similarly, a random variable $Y$ is said to be sub-exponential if $$
||Y||_{\psi_1}:=\sup _{q\geq 1}q^{-1}(\mathbb{E}|Y|^q)^{1/q}<\infty.$$
See for example  \cite{vershynin_2010}. Last, a random vector $Z=(Z_1, \dots, Z_p)$ is said to be sub-gaussian if $Z_i$ is sub-gaussian for every $i\in[p]$.

The following theorem provides us sufficient conditions on $\Sigma$, the underlying random variable \(Z\) and the number of samples $n$ to ensure, w.h.p., the approximate sparsity of $S$.

\begin{theorem}
\label{theorem:num_of_samples}
Assume $Z$ is a sub-gaussian random vector with covariance matrix $\Sigma$ which is $(r,\mu,R,2)$-sparse where $2R < \mu$. Let $K=\max_i||Z_i||_{\psi_2}$ and $t=\min\{\frac{\mu -2 R}{2\sqrt{p-r}+1},\frac{\mu - R}{4},K^{2}\}$. Then, for $n> C\frac{K^4}{t^2}\log (54p^2)$, where $C$ is an absolute constant, w.p. $\geq 2/3$, $S$ is $(r,\mu-t,\frac{1}{2}(\mu-t),2)$-sparse.  Moreover, with high probability, every large entry of $\Sigma$ (w.r.t. $\mu$) is a large entry of $S$ (w.r.t. $\mu-t$), and vice versa.
\end{theorem}

In addition to the probabilistic guarantees of Theorem \ref{theorem:num_of_samples} for fixed $p$ and $n$, we can also deduce stronger asymptotical results, as $n,p\rightarrow \infty$.

\begin{theorem}
\label{thm:assym_num_of_samples}
Assume $Z$ is a sub-gaussian random vector with covariance matrix $\Sigma$ which is $(r,\mu,R,2)$-sparse where $2R < \mu$, and let $t$ be as in Theorem \ref{theorem:num_of_samples}. Then, as $n\rightarrow \infty$ with $p$ fixed, or as $n,p\rightarrow \infty $ with $\frac{p\log p}{n} \rightarrow0$, the probability that $S$ is $(r,\mu-t,\frac{1}{2}(\mu-t),2)$-sparse with $J_{\mu - t}(S) = J_{\mu}(\Sigma)$ converges to one.
\end{theorem}

Note that for large \(p\gg 1\), the parameter \(t\) in Theorem \ref{theorem:num_of_samples} is equal to \((\mu-2R)/(2\sqrt{p-r+1})\). Hence, the required number of samples for $S$ to be approximately sparse is\ $n>O(\log(p)/t^2)=O(p\log p)$. With such a number of samples, we can replace the assumptions on $S$ with corresponding assumptions on $\Sigma$ and obtain the following result analogous to Theorem \ref{theorem:tree_prob}.

\begin{corollary}
Assume $\Sigma$ is $(r,\mu, R, 2)$-sparse, where $R < \mu/2$. Then, as $n,p\rightarrow \infty$ with $\frac{p\log p}{n} \rightarrow0$, with probability tending to 1, Algorithm \ref{algo:tree_algo} finds all large entries of $S$, which correspond to large entries of $\Sigma$, using at most $O(nrp\log ^2p)$ operations.
\end{corollary}

In various contemporary statistical applications, the number of samples $n$ is comparable to the dimension $p$. In such a case, we may still detect the large entries of the population covariance matrix in sub-quadratic time. For example, we can divide the \(p\) variables into $K=p^{1-\alpha}$ distinct groups, each of size $p^{\alpha}$, and separately find the largest entries in each of the \(K^2\) sub-matrices of the full covariance matrix. In each pair, Corollary 5.1 applies, since if \(p,n\to\infty \) with $p/n\to const$, then $p^\alpha\log p/n\to 0$. The overall run time, up to logarithmic factors in \(p\) is now higher \(O(np^{2-\alpha}r)\) but still sub-quadratic.

Finally, note that throughout this section we assumed that the underlying random variable \(Z\) is sub-Gaussian. It is an interesting question whether some of the above theorems continue to hold under weaker tail conditions. 

\section{Simulations}
\label{sec:simulations}

\begin{figure}[t]
    \includegraphics[scale=0.53]{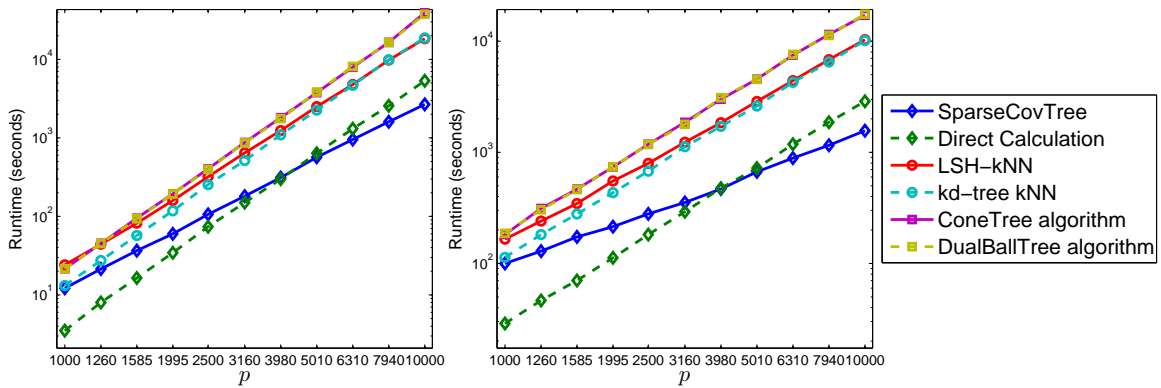}
\caption{Average runtime to detect large entries of a sparse covariance matrix as a function of the dimension for several algorithms, presented in a logartihmic scale in both axes. The input data was drawn from a Gaussian distribution with a random population covariance matrix and $r=\lfloor\frac{\log_2p}{3}\rfloor$. In the left panel, the number of samples increases with the dimension, $n=\lfloor p \log p\rfloor$, whereas in the right panel the number of samples is fixed, $n=50,000$.}
\label{fig:cov_runtime}
\end{figure}

In this section we illustrate the empirical performance of Algorithm 4, denoted SparseCovTree, on input drawn from a Gaussian distribution with a sparse covariance matrix. Furthermore, we compare it to other algorithms that can be used to locate the large entries of $S$: (1) LSH-$k$NN \cite{datar_indyk_2004}; (2) $kd$-tree \cite{kd_tree}; (3)
Dual-Ball-Ball (DBB) \cite{ram_2012}; (4) Dual-Ball-Cone (DBC) \cite{ram_2012}; and (5) direct calculation of all entries of \(S\), at a cost of \(O(np^{2})\) operations.
To the best of our knowledge, no public implementation of the sFFT algorithm of \cite{Nearly_Optimal_sFFT} is currently publicly available, thus we did not include it in the comparison.  For the LSH-$k$NN
and $kd$-tree algorithms, we used the mlpack library \cite{mlpack}, whereas for DBB and DBC algorithms we used the implementation kindly provided to us by  \cite{ram_2012}. All codes are in C++ and were compiled with the O3 optimization flag.

We generated  random population covariance matrices $\Sigma$ as follows: first $r=\lfloor\frac{\log_2p}{3}\rfloor$ entries in each row (and their corresponding transposed entries) were selected uniformly to be $\pm 1$, with all others set to zero. Then,  the diagonal entries were set to be $\pm 1$ as well. Last, to have a valid positive-definite covariance matrix, all diagonal entries were increased by the absolute value of the smallest eigenvalue of the resulting symmetric matrix plus one. Following this, we chose the threshold to be $\mu=0.5$.

The space complexity of SparseCovTree algorithm, which involves storing $m$ trees, is $O(npm)$. This may exceed the available memory for large $n, p$ and $m$. Thus, instead of simultaneously constructing and processing the entire \(m\) trees from the coarse level, we can construct and process  the sub-trees of each node in the different fine levels separately, starting from some coarse level $h\in[\log_{2} p]$. By doing so, the space complexity is reduced by a factor of $2^h$. These modifications can only increase the probability to locate large entries, whereas the theoretical runtime bound of Theorem \ref{theorem:tree_prob} increases to $O\left(nmp(\log_{2} p-h)2^{h}\right)$ operations.
 Particularly, in our experiments we used $h=5$, reducing the space complexity by a factor of $32$.
Moreover, the value  of $m$ as stated in Theorem \ref{theorem:tree_prob} is required mostly for the theoretical analysis. Since the multiplicative factor appearing in Lemma \ref{lemma:tree_test} is not sharp, and the probability of failure converges to zero as the dimension increases, we can in practice use a smaller number of trees $m$ and still detect most large entries of $S$,{ in sub-quadratic time}. Here we chose a fixed value of $m=20$, leading to discovery of more than $99\%$ of the large entries for all values of $p$ considered.

As for the LSH-$k$NN tuning parameters, we chose the number of projections to be 10, hash width 4 and bucket size 3500. For the second hash size, we picked a large prime number 424577, whereas the number of tables was configured according to \cite{LSH_lecture_notes_2008} with $\delta = \frac{1}{\log p}$. This configuration led to more than $99$\% discovery rate of all large entries, in all tested dimensions.

Figure \ref{fig:cov_runtime} shows the average runtime, in logarithmic scale, for various values of $p$ from 1000 to 10,000. Surprisingly, all alternative solutions, apart from  SparseCovTree, yield slower runtime performances than the direct calculation. We raise here several possible explanations for this. First, in all methods, except SparseCovTree and direct calculation, to locate negative large entries of $S$ we duplicate the input $\{\x_i\}$  with the opposite sign, thus potentially increasing the runtime by a factor of $2$. Moreover, it was suggested (see \cite{shrivastava_2014}) that space-partitioning-based methods, such as $kd$-tree,  DBB and DBC, may lead to slow runtime in high dimensions. For an empirical illustration of this issue, see \cite{Shrivastava_Li_2015}. In our case, the dimension of the search-problem is the number of samples $n$, which may indeed be very high. In contrast to theses algorithms, the LSH method was originally designed to cope well in high dimensions. However, to locate most of $O(pr)$ large entries of $S$, we tuned the LSH-$k$NN algorithm to have a small misdetection probability, by setting $\delta = \frac{1}{\log p}$. Such small value of $\delta$ led to a large number  tables, thus increasing the runtime by a significant factor.

 As for  SparseCovTree, in low dimensions direct calculation is preferable. However, for larger values of $p$, SparseCovTree is clearly faster. Moreover, the slope on a log-log scale of the direct approach is roughly two, whereas our method has  a slope closer to 1. 

In the above simulation, the underlying population matrix \(\Sigma\) was exactly sparse. Next, we empirically study the ability of the SparseCovTree algorithm to detect the large entries when the underlying population covariance matrix is not perfectly sparse, but only approximately so. Specifically, we generated a population covariance matrix similar to the procedure described above, only that the previously zero entries, are now $\pm\epsilon$. 
Figure \ref{fig:Sigma_epsilon} presents both the average run-time as well as the misdetection rate of our algorithm with \(m=10,25,50\) trees, as a function of $\epsilon$, for a covariance matrix of size \(p=2000\) and \(n=20000\) samples.
Each row of $\Sigma$ has an average of $r=8$ large entries, all of size $\mu=1$. Hence, \(R=R(\epsilon)\approx (p-r)\epsilon^2\). The critical value $\epsilon_{crit}$ at which $R=\mu/2$ is thus $\epsilon_{crit} = 1/\sqrt{4(p-r)}$. As seen in the left panel, the run-time scales roughly linearly with the number of trees, is nearly constant for \(\epsilon<\epsilon_{crit}\) and slowly increases for larger values of $\epsilon$. The right panel shows that in accordance to our theory, the misdetection rate is also nearly constant as long as \(\epsilon<\epsilon_{crit}\).
\begin{figure}[t]
\centering
\includegraphics[width=0.4\textwidth]{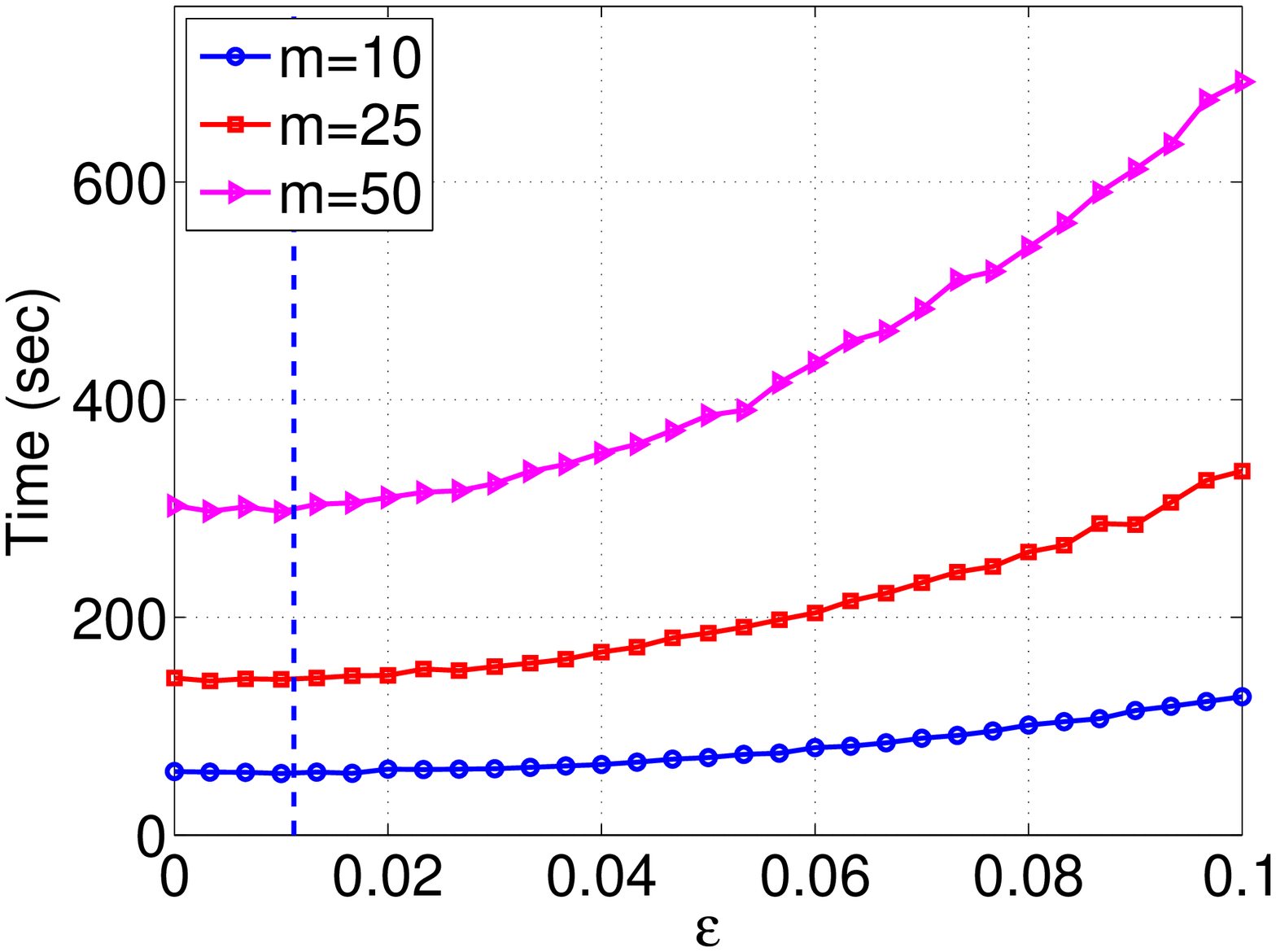}
\includegraphics[width=0.4\textwidth]{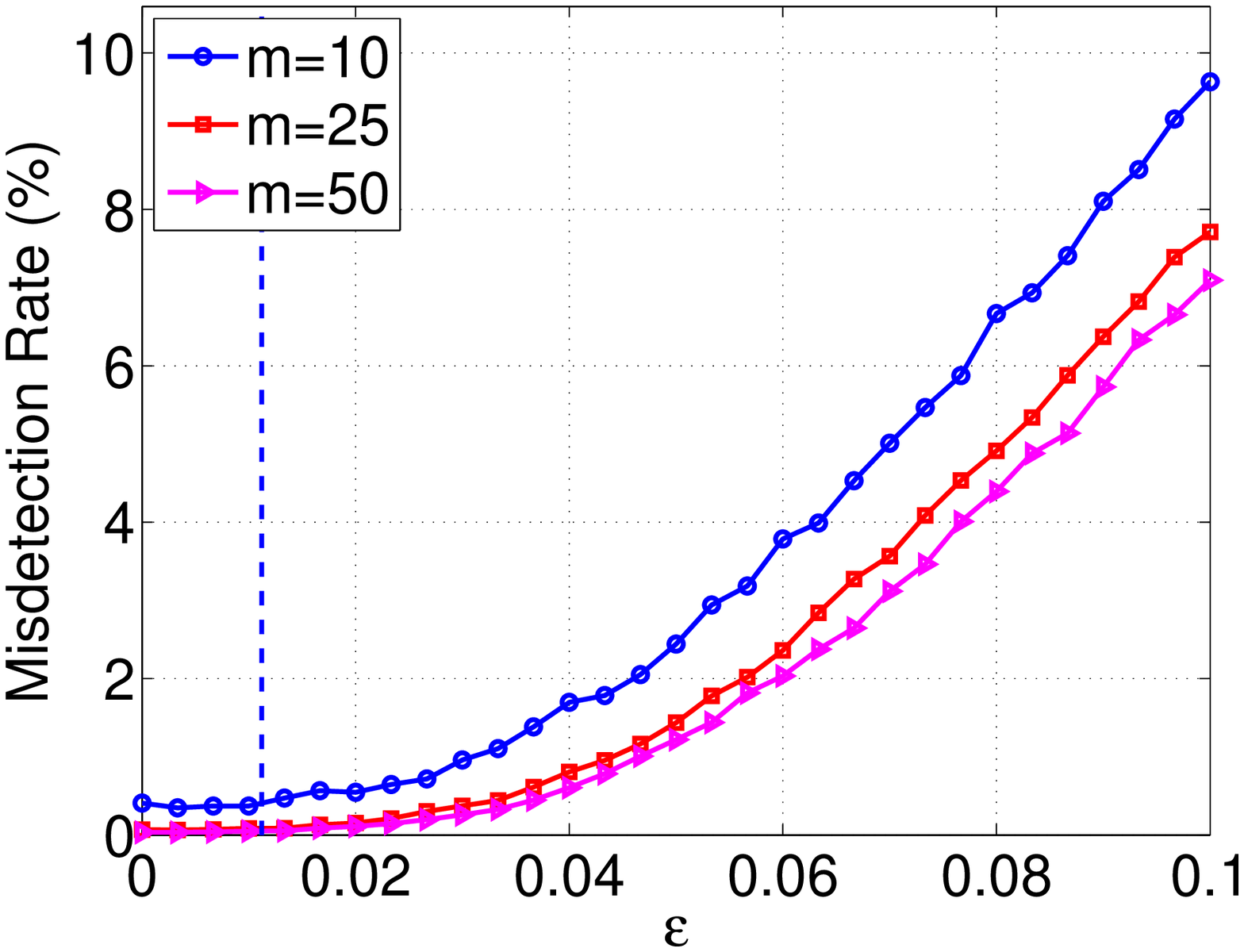}
\caption{Average run-time and misdetection rate (in percentage points) as a function of $\epsilon$, for fixed $p=2000$ and $n=20000$. The dashed vertical line is the critical value $\epsilon_{crit}$ where for the population matrix  \(R(\epsilon_{crit})=\mu/2.  \)}
\label{fig:Sigma_epsilon}
\end{figure}

Finally, we demonstrate an application of the SparseCovTree algorithm to detect near duplicates in  a large artificial data set: Given a reference data set with \(p\) elements, \(\{{\bf x}_i\}_{i=1}^p\), with each \(\x_i\in\mathbb{R}^n\), the goal is to quickly decide, for any query vector \({\bf y}\in\mathbb{R}^n\),  whether there exists $\x_i$ s.t. $Sim
(\x_i,\y)$ is larger than a given threshold $\mu$, where $Sim$ is the \textit{Cosine Similarity}. In contrast to the previous simulation, here the key quantity of interest is the average query time, and the computational effort of the pre-processing stage is typically ignored.

Specifically, we considered the following illustrative example: We first generated a reference set  $\{\x_i\}_{i=1}^p$, whose $p$ elements were all independently and uniformly distributed on the unit sphere $S^{n-1}$. Next, we evaluated the average run-time per query in two scenarios:\ a) a general query \(\y\), uniformly distributed on \(S^{n-1}\), and hence with high probability, not a near-duplicate; b) a near duplicate query $\bf{u}$, generated as follows:
$$
 {\bf u} = \sqrt{1-\epsilon}\x_{j} + \sqrt{\epsilon}\frac{\z - (\z^T \x_{j} )\x_{j}}{||\z - (\z^T \x_{j})\x_{j}||}
$$
where $\z$ is uniformly distributed on $S^{n-1}$ and the duplicate index $j$  is chosen uniformly from the set of elements $[p]$. 

In our simulations, the parameters $\epsilon = 0.25$, \(n=40,000\) and $m=20$ trees were kept fixed, and we varied the size  $p$ of the reference set. For a fair comparison to LSH, we decreased its number of projections to 7, and set its misdetection probability per single projection to be $\delta=0.3$.
Empirically, both our method and LSH\ correctly detected a near-duplicate with probability rate larger than 99\%. 
\begin{figure}[t!]
\centering
\includegraphics[scale=0.5]{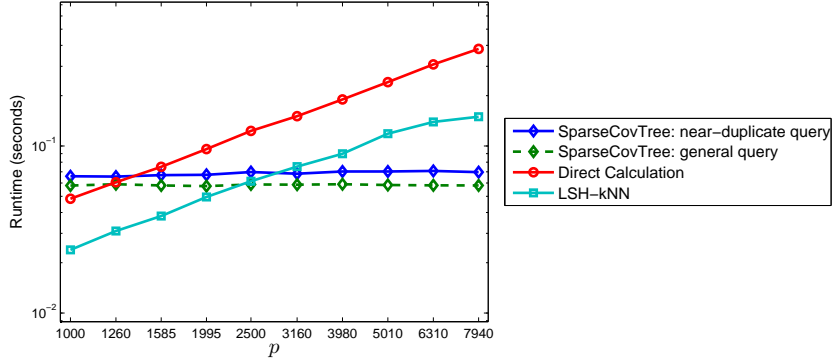}
\caption{Average runtime per single query to detect near-duplicates in a large artificial data set with $p$ elements  of dimension $n=40,000$, presented in a logarithmic scale in both axes. For SparseCovTree, we present the average runtime both for a near-duplicate and for a general input query. For the other two methods, this partition  did not affect the runtime significantly.}
\label{fig:near_dup}
\end{figure}

Figure \ref{fig:near_dup} shows the average query runtime for  SparseCovTree, LSH-$k$NN and the direct approach. In both LSH-$k$NN  and direct approach, the average query runtime did not depend on whether the query was a near-duplicate or not, and hence we show only one of them. For SparseCovTree, in contrast, the average runtime for a general query is lower  than for a near-duplicate query. The reason is that for a general query our  algorithm detects that this is not a near-duplicate by processing only the nodes at the first few top levels of the tree. For this problem, the LSH method is clearly preferable over direct calculation for all tested values of \(p\). When the number of elements \(p\) becomes sufficiently large, our method becomes faster. Most importantly, the runtime curves of SparseCovTree, both for the true and false near-duplicates, seem \textit{almost constant} compared to the other runtime curves, which are approximately linear in \(p\). This is   consistent with the theoretical analysis, as SparseCovTree query time  is $O(n\log^2p)$ operations, which is sub-linear, whereas the direct calculation query time is $O(np)$ operations, which is linear thus significantly larger, and LSH-$k$NN query time is $O(np^\rho\text{ poly} \log p)$ operations, for some constant $\rho > 0$, which is also sub-linear but not logarithmic as the SparseCovTree query time.

\section{Summary}

In this paper we considered the computational aspects of detecting the large entries of a sparse covariance matrix from given samples. We derived and theoretically analyzed two different algorithms for this task, under sparsity assumptions either on $S$ itself, or on $\Sigma$ and additional requirements on the underlying random variable $Z$ and number of samples $n$. We next demonstrated the time-efficiency of the algorithms, theoretically for both of them and empirically only for the tree-based one.

Our work raises several interesting questions for future research. If an oracle gave us the precise locations of all large entries, computing them directly would require $\Omega(npr)$ operations, a quantity lower by a $\text{poly}\log p$ factor from our results. This raises the following fundamental question: is it possible to reduce the poly-log factor, or is there a theoretical lower-bound on the required number of operations? In addition, in this article we only considered  a $L_2$-norm in the definition of matrix sparsity. It may be of interest to consider different norms, e.g. $L_1$, which could lead to different algorithmic results. Finally, both algorithms require partial knowledge on the sparsity parameters. It is thus of interest to efficiently estimate these parameters from the data, in case they are unknown.
\appendix

\section{Appendix - Proofs}
\subsection{Proof of Lemma \ref{lemma:sfft_prob}}
\begin{proof}For every independent run $i\in[m]$, by definition\((\y_i)_{j}=0\) for all \(j\notin J_i\). Hence,
\begin{equation} \label{eq:sfft_proof_1}||S_{k}-\y_i||^2 = \sum_{j\in J_i}(S_{kj} -(\y_i)_j)^2 + \sum_{j\not\in J_{i}}(S_{kj})^2 \geq \sum_{j\not\in J_{i}}(S_{kj} )^2 =||S_{k} -S_{k}^{J_i} ||^{2}
\end{equation}
where $(S_{k}^{J_i})_j=S_{kj}$ for $j\in J_i$ and otherwise $(S_{k}^{J_i})_j=0$. Moreover, for every $i\in[m]$,
$$||S_{k}||^2 - ||S_{k} -S_{k }^{J_{i}} ||^{2} =  \sum_{j\in J_{i}}(S_{kj})^2 \leq \sum_{j\in \bigcup_iJ_{i}}(S_{kj})^2 = ||S_{k}||^2 - ||S_{k} -\tilde S_{k } ||^2.$$
Hence
\begin{equation} \label{eq:sfft_proof_3}||S_{k} -\tilde S_{k }||\leq||S_{k} -S_{k }^{J_{i}} ||.
\end{equation}
By combining Eqs. (\ref{eq:sfft_proof_1}) and (\ref{eq:sfft_proof_3}), we obtain
\begin{equation} \label{eq:sfft_proof_4}
||S_{k} -\tilde S_{k }||\leq \min_{i\in[m]}||S_{k} -\y_{i} ||.
\end{equation}
Since \(\y_i\) is an output of the sFFT algorithm, according to Eq. (\ref{eq:sfft_algo_res}) it satisfies the following inequality, w.p. $\geq 2/3$
\begin{equation}\label{eq:sfft_proof_5} ||S_{k} -\y_{i} ||\leq  (1+\alpha)\min_{||\y||_0\leq r}||S_{k}-\y||+\delta||S_k||.
\end{equation}
Thus, by a union bound argument, w.p. $\geq1-(\frac{1}{3})^m=1-\frac{1}{3p}$ there exists at least one index $i$ for which Eq. (\ref{eq:sfft_proof_5}) holds. Combining this with Eq. (\ref{eq:sfft_proof_4}) concludes the proof.

\end{proof}

\subsection{Proof of Theorem \ref{claim:sfft_algo_correctness}}
\begin{proof}First we show that a sufficient condition for Algorithm
\ref{algo:sfft_algo} to detect all large entries of $S$ is that Eq. (\ref{eq:sfft_cor}) holds for all rows. Then, from Corollary \ref{cor:sfft_cor}, the probabilistic guarantee will follow. Using Eq. (\ref{eq:sfft_cor}) with $\delta = \frac{\epsilon}{R+\sqrt{r}M}$ and $\alpha =1$ yields
\begin{equation}
\label{eq:sfft_proof_2}
 ||S_{k}-\tilde S_{k}||\leq 2\min_{||\y||_0\leq r}||S_{k}-\y||+\frac{\epsilon}{R+\sqrt{r}M}||S_k||.
\end{equation}
Since $S$ is $(r, \mu, R, 2)$-sparse, then $||S_k||\leq R + \sqrt{r}M$,
which implies that
$$ ||S_{k}-\tilde S_{k}||\leq2\min_{||\y||_0\leq r}||S_{k}-\y||+\epsilon.$$
Let $T_\mu(S_k)$ be the $k$-th row of  $S$ thresholded at level $\mu$,
$$ (\text{T}_\mu(S_{k}))_{j}=S_{kj} \boldsymbol1[|S_{kj} |\geq\mu].$$
Each row of $S$ has at most $r$ large entries, thus $||T_\mu(S_k)||_0\leq r$ and from Eq. (\ref{eq:sfft_proof_2}) we obtain
$$ ||S_{k}-\tilde S_{k}||\leq 2||S_{k}-T_{\mu}(S_{k})||+\epsilon=2\left(\sum_{j\not\in J_\mu(S_k)}(S_{kj})^2\right)^{1/2}+\epsilon\leq 2R+\epsilon. $$
Now, assume to the contrary that there exists a large entry $S_{kj}$, for $j\in J_\mu(S_k)$, which the algorithm failed to detect. Then
$$ ||S_{k}-\tilde S_{k}||_{2}\geq |S_{kj}|\geq\mu$$
meaning
$ \mu \leq 2R+\epsilon$, which contradicts the assumption.

As for  the run-time of the algorithm,
the first two steps, calculating the matrix $W$ and the bound $M$, can be performed using only $O(np\log p)$ operations. Then,  for every row $k$, the algorithm makes $O(\log p)$ queries of the sFFT algorithm, where each query requires $O(nr\log p\log((R+\sqrt{r}M)p/\epsilon))$ operations. As for the last step, evaluating the output $\tilde S$ requires $O(n|I|)$ operations. Since the size of $I$ cannot exceed the number of sFFT queries multiplied by the number of operations for a single query, we conclude that the total number of operations is at most $O(nrp\log^2p\log((R+\sqrt{r}M)p/\epsilon))$.

\end{proof}

\subsection{Proof of Lemma \ref{lemma:tree_test}}
\begin{proof}First, let us study the quantity $\sigma^2$ and the distribution of $\hat \sigma ^2$, under both the null and under the alternative.

If $\mathcal{H}_0$ holds, since $S$ is $(r,\mu,R,2)$-sparse,$$ \sigma^2=\sum_{j\in  I(h,i)}\langle \x_{k}, \x_{j}\rangle^2\leq \sum_{j\not\in J_\mu(S_{k})}\langle \x_{k}, \x_{j}\rangle^2\leq R^{2}<\mu^{2}/4.$$
In contrast, if $\mathcal{H}_1$ holds, then there is at least one entry $|S_{kj}|\geq\mu$, leading to
$$ \sigma^2=\sum_{j\in I(h,i)}\langle \x_{k}, \x_{j}\rangle^2\geq \sum_{j\in J_\mu(S_{k}) \cap I(h,i)}\langle \x_{k}, \x_{j}\rangle^2\geq \mu^{2}.$$
Therefore, the two hypotheses
in Eq. (\ref{eq:tree_hyp}) may also be written as
$$ \mathcal{H}_0:\sigma^2 \leq R^{2} \quad \text{vs.} \quad \mathcal{H}_1:\sigma^2 \geq\mu^2. $$
Since for every $l\in[m]$, $y_l\sim N(0, \sigma^2)$, $${\hat\sigma}^2 \sim \frac{\sigma^2}{m} \chi^2_m.$$
In \cite{birnbaum_nadler_2012} and \cite{Johnstone_2009} respectively, the following two concentration inequalities for a \(\chi^2\) random variable were derived, for all $m\geq 2$ (see also \cite{Laurent_Massart_2000})
\begin{equation}\label{eq:chi_ineq1}
\Pr\left[\frac{\chi^2_m}{m}>1+\epsilon\right]\leq \frac{1}{\sqrt{\pi m} (\epsilon+\frac{2}{m})}\exp\left(-\frac{m}{2}(\epsilon -\log(1+\epsilon))\right) \qquad \forall\;\epsilon>0
\end{equation}
\begin{equation}\label{eq:chi_ineq2}
\Pr\left[\frac{\chi^2_m}{m}<  1-\epsilon\right]\leq \exp\left(- \frac{m\epsilon^2}{4}\right) \qquad\forall\;  0 <\epsilon<1.\end{equation}
We use these inequalities to bound the probabilities of false alarm and misdetection:
$$ \Pr[\textbf{false alarm}] = \Pr\left[{\hat\sigma}^2\geq \frac{3\mu^2}{4}\middle|\mathcal{H}_0\right]=
\Pr\left[\frac{\chi_m^2}{m}\geq \frac{3\mu^2}{4\sigma^2}\middle|\sigma^2\leq R^2\right]
\leq
\Pr\left[\frac{\chi_m^2}{m}\geq\frac{3\mu^2}{4R^{2}}\right]
$$
Since $4R^2 \leq \mu^2$,
$$\Pr[\textbf{false alarm}] \leq\Pr\left[\frac{\chi_m^2}{m}\geq3\right]. $$
Using (\ref{eq:chi_ineq1}) with $\epsilon = 2$, we obtain for $m \geq3\log(\frac{1}{\delta})$
$$ \Pr[\textbf{false alarm}] \leq\exp\left(-\frac{m}{2}(2-\log 3)\right)\leq\delta. $$
Similarly for the probability of misdetection,
$$ \Pr[\textbf{misdetection}] = \Pr\left[{\hat\sigma}^2< \frac{3\mu^2}{4}\middle|\mathcal{H}_1\right]=\Pr\left[\frac{\chi_m^2}{m}< \frac{3\mu^2}{4\sigma^2}\middle|\sigma^2\geq \mu^2\right]\leq \Pr\left[\frac{\chi_m^2}{m}< \frac{3}{4}\right].$$
Using (\ref{eq:chi_ineq2}) with $\epsilon = 1/4$, we obtain for $m \geq64\log(\frac{1}{\delta})$
$$
\Pr[\textbf{misdetection}] \leq\Pr\left[\frac{\chi_m^2}{m}< \frac{3}{4}\right]\leq  \exp\left(-\frac{m}{64}\right)\leq \delta.
$$

\end{proof}

\subsection{Proof of Theorem \ref{theorem:tree_prob}}
\begin{proof}We first introduce the following notation. For each $k\in [p]$ and $h\in[L]$, divide the nodes into two disjoint sets
$$ \mathcal{A}_k^h = \{ (h,i) :i\in[2^h],\; I(h,i) \cap J_{\mu}(S_k)\neq\emptyset\}$$
$$ \mathcal{B}_k^{h} = \{ (h,i) :i\in[2^h], \; I(h,i) \cap J_{\mu}(S_k)=\emptyset\}$$
We say that the algorithm visited or entered a node $T(h,i)$ while processing the trees with $\x_k$ if during its execution the call {\tt Find($\x_k, h, i, \{T_l\}, \mu$)} occurred.
Notice that as the algorithm processes the trees with input $\x_k$, in order to detect all large entries in the $k$-th row, while at level $h$ it \textit{must} visit all nodes in $\mathcal{A}_k^h$. In contrast, to have a fast runtime it should avoid processing nodes in $\mathcal{B}_k^h$.
Moreover, recall that the trees-construction stage deterministically requires $O(np\log p)$ operations. Thus, to prove the run-time bound, it suffices to show that after the preprocessing stage the algorithm uses at most $O(nrp\log ^2 p)$ operations, w.h.p.  for $(i)$ and $(ii)$, and in expectation for $(iii)$.

\textit{Proof of (i) and (ii).}
Let $V_{k}^{h}$ be the event that the algorithm visited all nodes in $\mathcal{A}_k^h$ but no nodes of $\mathcal{B}_k^h$, while processing the $h$-th level of the trees with $\x_k$. Since $S$ is sparse, $|\mathcal{A}_k^h|\leq r $, and therefore
$$\left| \bigcup_{k\in[p], h\in[L]}\mathcal{A}_k^h\right| \leq rpL.$$
Moreover, since the algorithm detects all large entries if and only if it visits all nodes of all sets $\mathcal{A}^h_k$, it suffices to bound the probability of $V_k^h$ simultaneously occurring for all $k\in[p]$ and $h\in[L]$.

 it has at most $2r$ new nodes to check in the $h+1$ level, unless $h = L$.
According  to lemma \ref{lemma:tree_test} with $\delta = \frac{1}{2rpL^3}$, the probability of the algorithm to enter a node in $\mathcal{B}_k^h$ or to skip a node in $\mathcal{A}_k^h$,  is at most $\frac{1}{2rpL^3}$. Denote by ${\overline V}_{k}^{h}$  the complementary event of ${V}_{k}^{h}$. This implies,
\begin{equation}
\label{eq:tree_proof_event}
\Pr\left[{\overline V}_{k}^{h} \right]\leq2r\cdot  \frac{1}{2rp L^{3}}=\frac{1}{pL^{3}}.
\end{equation}
Since
\begin{align*}\Pr\left[   { V}_{k}^{h} \right] &= \Pr\left[   {V}_{k}^{h} \middle| V_{k}^{1}, \dots, V_{k}^{h-1} \right]\Pr\left[V_{k}^{1}, \dots,V_{k}^{h-1} \right] \\&+ \Pr\left[    V_{k}^{h} \middle|\exists i\in[h-1]:{\overline V}_{k}^{i}   \right]\Pr\left[\exists i\in[h-1]:{\overline V}_{k}^{i} \right] \end{align*}
then
$$\Pr\left[   {V}_{k}^{h} \middle| V_{k}^{1}, \dots, V_{k}^{h-1} \right]= \frac{\Pr\left[{ V}_{k}^{h}  \right] -  \Pr\left[ V_{k}^{h} \middle| \;\exists i\in[h-1]:{\overline V}_{k} \right]\Pr\left[\exists i\in[h-1]:{\overline V}_{k}^{i} \right] }{\Pr\left[V_{k}^{1}, \dots,V_{k}^{h-1} \right]}.$$
Clearly $\Pr[V_{k}^{1}, \dots,V_{k}^{h-1}] \leq 1$, and thus
$$\Pr\left[   {V}_{k}^{h} \middle| V_{k}^{1}, \dots, V_{k}^{h-1} \right]\geq\Pr\left[   { V}_{k}^{h} \right] - \Pr\left[\exists i\in[h-1]:{\overline V}_{k}^{i} \right].$$
We use Eq. (\ref{eq:tree_proof_event}) to obtain
$$\Pr\left[   {V}_{k}^{h} \middle| V_{k}^{1}, \dots, V_{k}^{h-1} \right]\geq1-\frac{1}{pL^3} - (h-1)\frac{1}{pL^{3}} \geq1-\frac{1}{pL^{2}}.$$
Next, we derive, for a fixed $k$,
\begin{align*} \Pr\left[\forall h:V_{k}^{h}\right]&=\Pr\left[V_k^{0}\right]\cdot\prod_{h=1}^{L}\Pr\left[V_{k}^{h}\middle|V_k^0,\dots,V_{k}^{h-1}\right] \geq
\left( 1-\frac{1}{pL^{2}} \right)^{L}
\end{align*}
where clearly $\Pr[V_k^{0}]=1$. Therefore,

\begin{align}\label{eq:tree_prob_proof}
\Pr[\forall k, h: V_k^{h}] &=1 - \Pr[\exists k,h: {\overline V}_k^{h}] \geq1-\sum_{k=1}^p\Pr\left[\exists h :  {\overline V}_k^{h}\right]\\\nonumber&\geq 1 - p\left(1-\left( 1-\frac{1}{pL^{2}} \right)^{L}\right)
.\end{align}

To bound this expression, consider the function $g(x)=(1-x)^{L}$. Since $g''(\xi)\geq0$ for all $\xi\in[0,1]$,
it readily follows that for all $x\in[0,1]$, $g(x)  \geq (1-xL)$.
Hence, for $x=\frac{1}{pL^2}$ we obtain
\begin{equation}\label{eq:tree_proof_up_bound}
 \left(1-\frac{1}{pL^2}\right)^{L}\geq 1-\frac{1}{pL}.\end{equation}
Combining Eq. (\ref{eq:tree_prob_proof}) with Eq. (\ref{eq:tree_proof_up_bound}) yields
$$ \Pr[\forall k, h: V_k^{h}]\geq 1- p\left( 1- \left( 1-\frac{1}{pL }\right) \right)=1-\frac{1}{\log_{2} p}\xrightarrow{p\rightarrow\infty}1
$$
which proves \textit{(ii)}, and similarly, for $p\geq 8$,
$$ \Pr[\forall k, h: V_k^h] \geq 2/3$$
which proves \textit{(i)}.

\textit{Proof of (iii).}
Let $X_{h,i}^k$ be the indicator of the event that the algorithm entered the node $T(h,i)$, while processing the tree with $\x_k$. Recall that
$$ \Pr\left[X_{h,i}^k=1\right]\leq \Pr\left[{\hat\sigma}^2(I(h,i))\geq \frac{3\mu^2}{4}\right].$$
For nodes in $\mathcal{A}_k^h$, this probability is clearly at most one, whereas for nodes in $\mathcal{B}_k^h$ it is at most $\delta$, as $m\geq 64\log(\frac{1}{\delta})$. By its definition, the algorithm makes at most $cnm$ operations, for some absolute constant $c$, in each node it considers. Thus,
\begin{align*}
\mathbb{E}[\text{Runtime}]&=\mathbb{E}\left[\text{Number of visited nodes}\right]\cdot cmn \\&=
\sum_{k=1}^p\sum_{h=1}^{L}\sum_{i=1}^{2^h}\Pr\left[X_{h,i}^k=1\right]\cdot cmn\\&\leq
\sum_{k=1}^p\sum_{h=1}^{L}\left(|\mathcal{A}_k^h| + |\mathcal{B}_k^h| \delta\right)\cdot cmn
\\&\leq p\left( rL+2p\delta\right)\cdot cmn.\end{align*}
For $m\geq 64\log p$, $\delta=\frac{1}{p}$, and we conclude
$$ \mathbb{E}[\text{Runtime}]\leq C' nrp\log^2p $$
for some absolute constant $C'$.

\end{proof}

\subsection{Proof of Theorem \ref{theorem:num_of_samples}}

To prove Theorem \ref{theorem:num_of_samples} we first introduce the following auxiliary lemma. It provides exponential tail bounds on the deviations of a quadratic form from its mean value, when the underlying random vector has non-zero mean. 
This lemma thus generalizes previous concentration bounds for quadratic forms \cite{rudelson_2013} \cite{Boucheron_Lugosi_Massart_2013}, which all assumed the vector has mean zero. Hence, it may be of independent interest.

\begin{lemma}
\label{lemma:mat_ineq}
Let $Y=(Y_1, \dots, Y_n)^t$ be a sub-gaussian random vector where all $Y_i$'s are independent and let $K=\max_i ||Y_i||_{\psi_2}$. Let $A\in\mathbb{R}^{n\times n}$ be a real symmetric matrix. Then, for every $t>0$,
$$ \Pr\left[|Y^tAY-\mathbb{E}[Y^tAY]| > t\right]\leq 6 \exp\left[-c\min\left(\frac{t^2}{(n-1)K^4||A||^2_F}, \frac{t}{K^2||A||}\right)\right]$$
where $c$ is a positive absolute constant.
\end{lemma}
\begin{proof}
First, observe that due to the independent of the $Y_i$'s,
$$ Y^tAY-\mathbb{E}[Y^tAY] = \sum_i a_{ii}[Y_i^2-\mathbb{E}Y_i^2] + \sum_{i\neq j}a_{ij}[Y_{i}Y_j - \mathbb{E}Y_i\mathbb{E}Y_j].$$
Since
$$ (Y_i - \mathbb{E}Y_i)(Y_j - \mathbb{E}Y_j) = Y_iY_j - Y_j\mathbb{E}Y_i-Y_i\mathbb{E}Y_j + \mathbb{E}Y_i\mathbb{E}Y_j$$
and $A$ is symmetric, we obtain
\begin{align*}  \sum_{i\neq j}a_{ij}[Y_iY_j - \mathbb{E}Y_i\mathbb{E}Y_j]&= \sum_{i\neq j}a_{ij}(Y_i - \mathbb{E}Y_i)(Y_j - \mathbb{E}Y_j)+2\sum_{i\neq j}a_{ij}[Y_i\mathbb{E}Y_j - \mathbb{E}Y_i\mathbb{E}Y_j]\\
& =\sum_{i\neq j}a_{ij}(Y_i - \mathbb{E}Y_i)(Y_j - \mathbb{E}Y_j)+\sum_{i}b_i[Y_i - \mathbb{E}Y_i]\end{align*}
where $b_i = 2\sum_{j\neq i}a_{ij}\mathbb{E}\y_j$. Therefore,
\begin{align*} \left|Y^tAY-\mathbb{E}[Y^tAY]\right|\leq&  \left|\sum_i a_{ii}[Y_i^2-\mathbb{E}Y_i^2]\right|+\left|\sum_{i\neq j}a_{ij}(Y_i - \mathbb{E}Y_i)(Y_j - \mathbb{E}Y_j)\right|+\left|\sum_{i}b_i[Y_i - \mathbb{E}Y_i]\right|
\end{align*}
and the problem reduces to estimating the following probabilities
\begin{align}
\label{eq:lemma_proof_hw_prob}
\rho:= \Pr\left[|Y^tAY-\mathbb{E}[Y^tAY]| > t\right]\leq & \Pr\left[|\sum_i a_{ii}[Y_i^2-\mathbb{E}Y_i^2]| > t/3\right]+\\
& \nonumber\Pr\left[|\sum_{i\neq j}a_{ij}(Y_i - \mathbb{E}Y_i)(Y_j - \mathbb{E}Y_j)| > t/3\right]+\\
& \nonumber\Pr\left[|\sum_{i}b_i[Y_i - \mathbb{E}Y_i]| > t/3\right]=:\rho_1+\rho_2+\rho_3.
\end{align}
As for the first term, notice that $Y_i^2 -\mathbb{E}Y_i^2$'s are independent centered sub-exponentials, with
$$ ||Y_i^2 - \mathbb{E}Y_i^2||_{\psi_1}\leq 4 ||Y_i||_{\psi_2}^2\leq 4K^2.$$
Thus we can use Proposition 5.16 of \cite{vershynin_2010} to obtain
$$ \rho_1 \leq 2\exp\left[-c_1 \min \left(\frac{t^2}{K^4\sum_{i}a_{ii}^2}, \frac{t}{K^2\max|a_{ii}|}\right)\right]$$
for some absolute constant $c_1$. Next, to bound the second sum, we follow the arguments of \cite{rudelson_2013} in the proof of the Hanson-Wright inequality to obtain
$$ \rho_2 \leq 2\exp\left[-c_2 \min \left(\frac{t^2}{K^4||A||^2_{F}}, \frac{t}{K^2||A||}\right)\right]$$
for some absolute constant $c_2$.

If $\mathbb{E}[Y]=0$ or $A$ is a diagonal matrix, then all coefficients $b_i=0$. The third term in Eq. (\ref{eq:lemma_proof_hw_prob}) then vanished, and we recover a results similar to the original Hanson-Wright inequality.
 Assume then that $\mathbb{E}Y \neq 0$ and $A$ is not diagonal, which leads to $\bf b \neq 0$. Since the random variables $Y_i - \mathbb{E}Y_i$'s are centered independent sub-gaussians with $||Y_i-\mathbb{E}Y_i||_{\psi_2}\leq ||Y_i||_{\psi_2}+|\mathbb{E}Y_i|\leq2K$, we then use Proposition 5.10 of \cite{vershynin_2010} to obtain
$$ \rho_3 \leq 2\exp\left[-c_3\frac{t^2}{K^2||{\bf b}||_2^2}\right]$$
for some absolute constant $c_3$.
 Since $||{\bf b}||_2^2\leq (n-1)K^2||A||^2_F$,
$ \sum_{i}a_{ii}^2\leq ||A||^2_F$
and $\max|a_{ii}|\leq ||A|| $
the lemma follows.

\end{proof}

\begin{proof}[Proof of Theorem \ref{theorem:num_of_samples}]
Let $A$ be the $n\times n$ matrix
$$A = \frac{1}{n-1}\left[I_{n}-\frac{1}{n}{\bf e}{\bf e^t}\right] $$
where ${\bf e}=(1, \dots, 1)\in\mathbb{R}^n$. Since
$$ S_{ij}=\z_i^tA\z_j = \frac{1}{2}\left[\z_i^tA\z_i + \z_j^tA\z_j - (\z_i - \z_j)^tA(\z_i - \z_j) \right]$$
and $\Sigma_{ij}=\mathbb{E}[S_{ij}]$, then
\begin{align*}
\Pr\left[|S_{ij}-\Sigma_{ij}|> t\right]\leq& \Pr\left[|\z_i^tA\z_i -\mathbb{E}[\z_i^tA\z_i] |>\frac{2t}{3}\right]+\\
& \Pr\left[|\z_j^tA\z_j -\mathbb{E}[\z_j^tA\z_j] |>\frac{2t}{3}\right]+\\
& \Pr\left[|(\z_i - \z_j)^tA(\z_i - \z_j) -\mathbb{E}[(\z_i - \z_j)^tA(\z_i - \z_j)] |>\frac{2t}{3}\right].
\end{align*}

Now, since $\z_i$ and $\z_i-\z_j$ are sub-gaussian vectors with independent entries, we use Lemma \ref{lemma:mat_ineq} to obtain
$$ \Pr\left[|S_{ij}-\Sigma_{ij}|> t\right]\leq18\exp\left[-c\min\left(\frac{t^2}{(n-1)K^4||A||^2_F}, \frac{t}{K^2||A||}\right)\right] $$
for some absolute constant $c>0$. Clearly both $||A||_F$ and $||A||$ are bounded by $\frac{2}{n-1}$, and therefore we conclude that for $ t\leq K^2$
\begin{equation}
\label{eq:samples_HW_ineq}
\Pr\left[|S_{ij}-\Sigma_{ij}|> t\right]\leq 18\exp\left( -\frac{Ct^2(n-1)}{K^4} \right).
\end{equation}
for some absolute constant $C$.
For $n>\frac{CK^4}{t^{2}}\log (54p^2)$, this probability is smaller than $\frac{1}{3p^2}$. Then, applying a simple union-bound argument yields
\begin{equation}
\label{eq:samples_event}
\Pr\left[\forall i,j: |S_{ij}-\Sigma_{ij}|\leq t\right]\geq \frac{2}{3} .
\end{equation}

Next, assuming the event in Eq. (\ref{eq:samples_event}) holds, we show that for sufficiently small $t$, sparsity of $\Sigma$ implies sparsity of $S$. In particular, for $t\leq\min\{\frac{\mu -2 R}{2\sqrt{p-r}+1},\frac{\mu - R}{4}\}$, we show that $S$ is $(r, \mu - t,\frac{1}{2}(\mu-t), 2)$-sparse,
and moreover, that $J_\mu(\Sigma) = J_{\mu-t}(S)$.

 Indeed, if $|\Sigma_{kj}|\geq\mu$ then according to Eq. (\ref{eq:samples_event}),$$|S_{kj}|\geq |\Sigma_{kj}|-|\Sigma_{kj}-S_{kj}|\geq\mu - t$$
meaning $J_\mu(\Sigma_k) \subseteq J_{\mu-t}(S_k)$.
Similarly, if $|S_{kj}|\geq \mu - t$ then
$$|\Sigma_{kj}|\geq |S_{kj}|-|\Sigma_{kj}-S_{kj}|\geq\mu -2 t\geq\mu -\frac{\mu - R}{2}=\frac{\mu + R}{2}> R.$$
Since every element of $\Sigma$ (in absolute value) is either larger than $\mu$ or smaller than $R$, the condition \(|S_{kj}|>\mu-t\) thus implies that $|\Sigma_{kj}|\geq \mu$ and hence $J_\mu(\Sigma_k) \supseteq\ J_{\mu-t}(S_k)$.
Combining these two results above yields that w.h.p. $J_\mu(\Sigma_k) = J_{\mu-t}(S_k)$,  and entries of \(S\) are large if and only if the corresponding entries of  $\Sigma$ are  large.

As for the gap,
\begin{align*}
\sum_{j\not\in J_{\mu-t}(S_k)}(S_{kj})^2\ &\leq   \sum_{j\not\in J_{\mu-t}(S_k)}(|\Sigma_{kj}|+t)^2=
\sum_{j\not\in J_{\mu}(\Sigma_k)}(|\Sigma_{kj}|+t)^2
\\
& \leq \sum_{j\not\in J_{\mu}(\Sigma_k)}(\Sigma_{kj})^2+\sum_{j\not\in J_{\mu}(\Sigma_k)}t^2+2\sum_{j\not\in J_{\mu}(\Sigma_k)}|\Sigma_{kj}|t\\
&\leq R^2+(p-r)t^2+2t\sqrt{p-r}R=(R + \sqrt{p-r}t)^2.
\end{align*}
For $t \leq \frac{\mu-2R}{2\sqrt{p-r}+1}$, we obtain $$\sum_{j\not\in J_{\mu-t}(S_{k})}(S_{kj})^2\leq (\frac{1}{2}(\mu-t))^2$$
which implies that $S$ is $(r, \mu-t, \frac{1}{2}(\mu-t),2)$-sparse.

\end{proof}

\subsection{Proof of Theorem \ref{thm:assym_num_of_samples}}
\begin{proof}
For any $n$ and $p$, let $\beta(n,p)$ be s.t. $n=\frac{2K^4}{ct^2}\log (\beta(n,p)p^2)+1$.
Notice that for $n$ and $p$ large enough, $\beta(n,p)>0$. By applying a union-bound argument on Eq. (\ref{eq:samples_HW_ineq}) we obtain, instead of Eq. (\ref{eq:samples_event}),
$$ \Pr\left[\forall i,j: |S_{ij}-\Sigma_{ij}|\leq t\right]\geq 1-\frac{1}{\beta(n,p)}. $$
Clearly as $n\rightarrow \infty$ with  $\frac{p\log p}{n}\rightarrow 0$, $\beta(n,p)\rightarrow\infty$, which  implies that this probability converges to one.

\end{proof}

\section*{Acknowledgments}

We thank Robert Krauthgamer and Ohad Shamir for interesting discussions.
We also thank Parikshit Ram and Alexander Gray for providing us with the code of their work and for various additional references.
Finally, we also thank Rasmus Pagh and Anshumali Shrivastava for several references. This work was partially supported by a grant from the Intel
Collaborative Research Institute for Computational Intelligence (ICRI-CI).


\clearpage
\bibliographystyle{unsrt}
\bibliography{bibfile}
%


\end{document}